\documentclass[11pt]{article}
\usepackage{amsmath, amsthm, amsfonts, mathtools}
\usepackage{tikz}
\usepackage{graphics,graphicx}
\usepackage[numbers]{natbib}

\usepackage{times}
\usepackage{amssymb}
\usepackage{enumerate}
\usepackage{bm}

\usepackage{tikz}

\usepackage{amsfonts}

\usepackage{amsthm, amsmath,amsfonts, amssymb}

\newcommand{\ignore}[1]{{}}

\def\p1{\phantom{1}}

\newtheorem{theorem}{Theorem}
\newtheorem{corollary}{Corollary}
\newtheorem{lemma}{Lemma}

\newtheorem{proposition}{Proposition}
\newtheorem{definition}{Definition}

\newcommand{\full}[1]{{}}

\begin{document}

\title{Motivating Time-Inconsistent Agents: A Computational Approach}
\author{Susanne Albers \thanks{Department of Computer Science, Technische Universit\"at M\"unchen, 85748 Garching, 
Germany; {\tt albers@in.tum.de}. Work supported by the German Research Foundation, project Al 464/7-1.}\and 
Dennis Kraft\thanks{Department of Computer Science,Technische Universit\"at M\"unchen, 85748 Garching, Germany. 
{\tt kraftd@in.tum.de} }}
\date{}
\maketitle


\begin{abstract}
In this paper we investigate the computational complexity of motivating time-inconsistent agents to complete long term projects.
We resort to an elegant graph-theoretic model, introduced by Kleinberg and Oren~\cite{KO}, which consists of a task graph $G$ with $n$ vertices, including a source $s$ and target $t$, and an agent that incrementally constructs a path from $s$ to $t$ in order to collect rewards.
The twist is that the agent is present-biased and discounts future costs and rewards by a factor $\beta\in [0,1]$.
Our design objective is to ensure that the agent reaches $t$ i.e.\ completes the project, for as little reward as possible.
Such graphs are called motivating.
We consider two strategies.

First, we place a single reward $r$ at $t$ and try to guide the agent by removing edges from $G$.
We prove that deciding the existence of such motivating subgraphs is NP-complete if $r$ is fixed.
More importantly, we generalize our reduction to a hardness of approximation result for computing the minimum $r$ that admits a motivating subgraph.
In particular, we show that no polynomial-time approximation to within a ratio of $\sqrt{n}/4$ or less is possible, unless ${\rm P}={\rm NP}$.
Furthermore, we develop a $(1+\sqrt{n})$-approximation algorithm and thus settle the approximability of computing motivating subgraphs.

Secondly, we study motivating reward configurations, where non-negative rewards $r(v)$ may be placed on arbitrary vertices $v$ of $G$.
The agent only receives the rewards of visited vertices.
Again we give an NP-completeness result for deciding the existence of a motivating reward configuration within a fixed budget $b$.
This result even holds if $b=0$, which in turn implies that no efficient approximation of a minimum $b$ within a ration grater or equal to $1$ is possible, unless ${\rm P}={\rm NP}$.
\end{abstract}


\section{Introduction}
Motivated by a recent paper of Kleinberg and Oren~\cite{KO}, we study the phenomenon of {\em time-inconsistent behavior\/} from a computer science perspective.
This fundamental problem in behavioral economics has many examples in every day life, including academia.
Consider, for instance, a referee who agrees to evaluate a scientific proposal.
Despite good intentions, the referee gets distracted and never submits a report.
Or consider a student who enrolls in a course.
After successfully completing the first couple of homework assignments the student drops out without earning any credit points.
In general, these situations have a reoccurring pattern.
An agent makes a plan to complete a set of tasks in the future but changes the plan at a later point in time.
Sometimes this is the result of unforeseen circumstances.
However, in many cases the plan is changed or abandoned even if the circumstances are the same as when the plan was made.
This paradox behavior of {\em procrastination\/} and {\em abandonment\/} is well-known in the field of behavioral economics and can have substantial effects on the performance of agents in an economic or social domain, see e.g.~\cite{A,DR1,DR2}.

A sensible explanation for time-inconsistent behavior is that agents assign disproportionally greater value to current cost than to future expenses. 
As an example, consider a simple {\em car wash problem\/} in which an agent, say Alice, is promised extra pocket money for washing her family's car.
Each day Alice can either do the chore or postpone it to the next day.
However, the longer she waits, the dirtier the car gets.
Assume that washing the car on day $i$, with $i\geq 1$, incurs a cost of $i/50$.
The cost of waiting another day is~$0$.
After completion of the task, she will receive a reward of $1$ Euro from the family. 
Alice is present-biased, i.e.\ she perceives current cost according to its true value but discounts future costs and rewards by a factor of $\beta \in [0,1]$. 
On day $i$ she compares the cost of washing the car right away, which is $i/50$, to the perceived cost of washing it on the next day, which is $\beta (i+1)/50$.
Suppose that $\beta = 1/3$. 
Because $i/50>\beta(i+1)/50$, she procrastinates with good intentions of doing the job on the following day.
On day $i=50$ Alice realizes that her perceived cost for washing the car on the next day, or any of the following days, is at least $\beta (50+1)/50$, which exceeds the perceived value of $\beta$.
Thus she abandons the project altogether.

{\bf Previous work:} In the economic literature there exists a considerable body of work on time-inconsistent behavior, cf.\ again~\cite{A,DR1,DR2}. 
We build on work by Kleinberg and Oren~\cite{KO} who propose a graph-theoretic model that elegantly captures the phenomena of procrastination and abandonment as observed in the car wash problem.
We will formally define the model in Section~\ref{sec:model}.
Essentially, it consists of a directed acyclic graph $G$ with $n$ vertices that models a long term project.
An agent, with bias factor $\beta\in[0,1]$, incrementally constructs a path from a designated source $s$ to a target $t$.
The edges of $G$, which represent the individual tasks of the project, are assigned non-negative costs.
The vertices, except for $s$ and $t$, correspond to intermediate states of the project.
When located at a vertex $v$, the agent chooses a path $P$ from $v$ to $t$ that minimizes its {\em perceived cost\/}. 
This means, the agent accounts for the first edge of $P$ by its true cost, whereas all remaining edges are discounted by $\beta$.

Kleinberg and Oren~\cite{KO} first investigating structural properties of $G$ under the assumption that the agent must not abandon the project.
In particular, they characterize task graphs in which the ratio between the total cost of a path traversed by the agent and the minimum cost of a path from $s$ to $t$ is exponential in $n$.
Interestingly, any such graph, after removing the direction of its edges, must contain a {\em $k$-fan\/}, with $k \in \mathcal{O}(n)$, as a minor.
Furthermore, Kleinberg and Oren analyze the number of different paths traversed the agent as $\beta$ varies between~$0$ and $1$.
They show that this number is in $\mathcal{O}(n^2)$.
The later result requires a consistent tie-breaking rule should the agent be indifferent between outgoing edges of the same vertex.

Next, Kleinberg and Oren~\cite{KO} assume that the agent is free to abandon the project and place a reward $r$ at $t$ as to motivate it to finish.
When located at $v$, the agent continues to follow a path that minimizes its perceived cost as long as it does not exceed the perceived value of the reward.
A graph in which the agent always successfully traverses a path from $s$ to $t$ is called {\em motivating\/} .
Kleinberg and Oren~\cite{KO} are interested in finding motivating subgraphs, by removing edges from $G$.
The authors present structural properties that any motivating subgraph with a minimal number of edges must satisfy. 

Finally, Kleinberg and Oren~\cite{KO} point to a number of open computational problems, including the complexity of finding motivating subgraphs.
Moreover, Kleinberg and Oren propose a problem setting in which rewards may be placed at intermediate vertices instead of just $t$.
In this case $G$ may not be pruned.
We call such a construct a {\em reward configuration}.

In an unpublished manuscript Tang et al.~\cite{TT} address some of the open problems. 
They refine the result on the cost ratio, which relates the path traversed by the agent to a cheapest path.
More specifically, they show that any task graph that does not contain a $k$-fan as minor after removing the direction of its edges must have a cost ratio of at most $\beta^{2-k}$.
Hence, for any fixed $k$, the cost ratio is constant in $n$.
Moreover, Tang et al. prove that it is NP-hard to decide if $G$ contains a motivating subgraph for a fixed reward.
Finally, they explore the problem of deciding the existence of a motivating reward configuration within a given budget.
Tang et al. give NP-hardness results for three variations of the problem.
They distinguish between configurations that are restricted to non-negative rewards and configurations that allow for any real-valued rewards.
Furthermore, they consider a setting in which every reward that is laid out on $G$ must also be collected.
In each of the three problem variations, Tang et al. measure the total value of a configuration by the sum of the absolute values of all rewards placed on $G$.

{\bf Our contribution:} In this paper we focus on the complexity and approximability of finding motivating subgraphs and reward configurations.
Our objective is budget efficiency.
Note that we take a design perspective.
In particular, we are not interested in minimizing the total cost experienced by the agent on its walk from $s$ to $t$ but rather the reward necessary to motivate the agent.

As for the first problem, by removing edges from $G$, it is possible to limit the agent's options in each of its steps.
Ideally, this prevents the agent from pursuing costly distractions and thereby reduces the reward required for it to finish the project.
The benefit of choice reduction is a well-known phenomenon in the field of behavioral economics.
It also has a very natural intuition in many real-live projects.
Take for instance the car wash problem.
As we will show in Section~\ref{sec:model}, the removal of edges in the problem's task graph corresponds to the introduction of deadlines.

The second problem takes a slightly less restrictive approach and allows the placement of intermediate rewards at arbitrary vertices of $G$.
Again this is meant to prevent the agent from pursuing distractions and encourage it to complete the project. 
We examine a version of the problem that, in our view, is the most sensible one. 
First, only non-negative rewards may be laid out.
This assumption is reasonable as it could be hard to convince an agent to pursue projects in which it has to make payments.
Furthermore, it is not clear how to account for such payments in the budget.
Secondly, the cost of a reward configuration is only measured by the sum of the rewards that are placed at vertice visited by the agent on its walk from $s$ to $t$.
This setting is a fundamentally different from the ones analyzed by Tang et al. as it may lead to configurations in which the agent is motivated by rewards that are never claimed.
Such configurations are also called {\em exploitative}.
We give an example in Section~\ref{sec:model}.

In Section~\ref{sec:comp} we settle the complexity of finding a motivating subgraph for a fixed $r$. 
We first observe that the problem is polynomially solvable if $\beta=0$ or $\beta=1$. 
We then prove that, for general $\beta\in(0,1)$, it is NP-complete to decide the existence of a motivating subgraph. 
In their paper~\cite{TT}, Tang et al. showed NP-hardness via a reduction from 3-SAT.
In contrast, we present a different reduction via $k$ DISJOINT CONNECTING PATHS~\cite{GJ}.
We believe that this reduction is slightly simpler. 
More importantly, we are able to generalize the reduction and show a hardness of approximation result in the following section.

In Section~\ref{sec:approx} we study the optimization version of the motivating subgraph problem. 
More formally, given a $\beta\in(0,1)$, determine the smallest possible value of $r$ such that $G$ contains a motivating subgraph.
We develop a $(1+\sqrt{n}$)-approximation algorithm that outputs $r$ as well as a corresponding motivating subgraph.
Interestingly, these subgraphs are paths. 
The algorithm is in fact a combination of two strategies, one which computes good solutions for small $\beta$ and one which is effective for large $\beta$. 
Furthermore, the approximation factor of our algorithm is asymptotically tight. 
As the main technical contribution of this paper, we prove that the optimization problem cannot be approximated in polynomial-time to within a ratio of $\sqrt{n}/4$ or less unless ${\rm P}={\rm NP}$.
Thus we resolve the approximability of the problem.

In Section~\ref{sec:ir} we explore the problem of finding reward configuration within a fixed total budget of at most $b$.
We show that the problem can again be solved in polynomial-time if $\beta=0$ or $\beta=1$.
Using a reduction from SET PACKING~\cite{GJ}, we prove that deciding the existence of a motivating reward configuration is NP-complete for general $\beta\in(0,1)$, even if $b=0$. 
This immediately implies that the optimization problem of computing the minimum $b$ that admits a motivating reward configuration cannot be approximated efficiently to within any ratio greater or equal to $1$ unless ${\rm P}={\rm NP}$. 


\section{The formal model}\label{sec:model}
In the following, we present the model by Kleinberg and Oren~\cite{KO}. 
Let $G =(V,E)$ be a directed acyclic graph.
Associated with each edge $(v,w)$ is a non-negative cost $c_G(v,w)$. 
An agent, with a bias factor $\beta\in[0,1]$, has to incrementally construct a path from a source $s$ to a target $t$.
At any vertex $v$ the agent evaluates its {\em lowest perceived cost\/}. 
For this purpose, the agent considers all paths from $v$ to $t$ and accounts for the cost of incident outgoing edges by their actual value, whereas it discounts future edges by $\beta$.
More specifically, let $d_G(w)$ denote the cost of a cheapest path from some vertex $w$ to $t$, considering the original edge costs.
If no path exists, we assume that $d_G(w)=\infty$.
Accordingly, the agent's lowest perceived cost is defined as ${\zeta_G(v)=\min\{c_G(v,w) + \beta d_G(w) \mid (v,w) \in E\}}$ if $v$ has at least one incident outgoing edge.
Otherwise we assume that $\zeta_G(v) = \infty$.
The agent only traverses outgoing edges $(v,w)$ for which ${c_G(v,w) + \beta d_G(w)=\zeta_G(v)}$.
Ties are broken arbitrarily.
Should $G$ be clear from context, we will omit the index and write $c(v,w)$, $d(v)$ and $\zeta(v)$ instead.

In Section~\ref{sec:comp} and~\ref{sec:approx} we will investigate problems in which a single non-negative reward $r$ is placed at $t$.
The agent perceives the value of this reward as $\beta r$ at every vertex different from $t$.
A graph $G$ is {\em motivating\/} if the agent does not abandon the project while constructing a path from $s$ to $t$ in $G$.
More specifically, at any vertex $v$ along the agent's path, it compares $\zeta(v)$ to $\beta r$ and continues moving if $\zeta(v) \leq \beta r$, i.e.\ the reward is sufficiently motivating.
Otherwise, if $\zeta(v) > \beta r$, the agent abandons.
Because ties are broken arbitrarily, there could be more than one path for the agent.
Consequently, $G$ is only considered motivating if the agent abandons the project on non of these paths. 

In Section~\ref{sec:ir} we will generalize Kleinberg and Oren's model to allow the placement of non-negative rewards $r(v)$ at arbitrary vertices $v$.
We call such a placement a {\em reward configuration\/}.
Given a specific reward configuration $r$, let $c_r(v,w) = c(v,w) - r(w)$ be the cost of traversing $(v,w)$ minus the reward collected at $w$ with respect to $r$.
Using $c_r$ as new cost metric, we denote the cost of a cheapest path from $w$ to $t$ as $d_r(w)$.
When located at $v$, the agent considers all paths from $v$ to $t$ and accounts for incident outgoing edges by their actual value, whereas future costs and rewards are discounted by $\beta$.
More specifically, we define the agent's perceived cost as $\zeta_r(v)=\min\{c(v,w) + \beta(d_r(w) - r(w)) \mid (v,w)\in E\}$.
The agent continues moving if $\zeta_r(v) \geq 0$. 
In this case the agent traverses an outgoing edge $(v,w)$ which minimizes its perceived cost, i.e.\ $c(v,w) + \beta(d_r(w) - r(w))=\zeta_r(v)$.
Again ties are broken arbitrarily.
Otherwise, if $\zeta_r(v) > 0$, the agent abandons. 
The agent only collects the rewards that it visits on its path from $s$ to $t$.
We are only interested in the value of the total reward handed out to the agent.
We say $r$ is within some given budget $b$ if the agent does not collect a total reward greater than $b$ on any of these paths. 
 
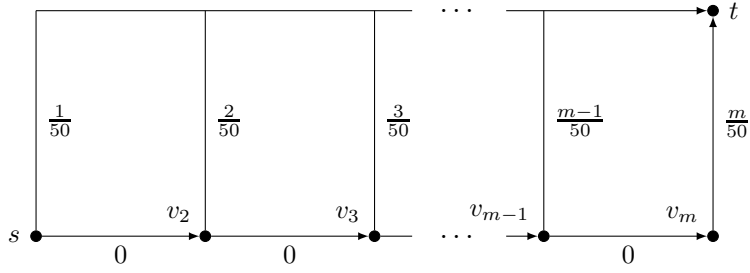
\begin{figure}[t]
\center
	\begin{tikzpicture}[nst/.style={draw,circle,fill=black,minimum size=4pt,inner sep=0pt]}, est/.style={draw,>=latex,->}]
			
		\node[nst] (v1) at (0,0) [label=left:\small{$s$}] {};
		\node[nst] (v2) at (2.25,0) [label=above left:\small{$v_2$}] {};
		\node[nst] (v3) at (4.5,0) [label=above left:\small{$v_3$}] {};
		\node[nst] (v4) at (6.75,0) [label=above left:\small{$v_{m-1}$}] {};
		\node[nst] (v5) at (9,0) [label=above left:\small{$v_{m}$}] {};
		\node[nst] (v6) at (9,3) [label=right:\small{$t$}] {};

		\node at (5.625,0) {$\dots$};
		\node at (5.625,3) {$\dots$};	
		
		\path (v1) edge[est] node [below] {\small{$0$}} (v2)
		(v2) edge[est] node [below] {\small{$0$}} (v3)
		(v3) edge (5,0)
		(6.25,0) edge[est] (v4)
		(v4) edge[est] node [below] {\small{$0$}} (v5)
		(v1) edge node [right] {\small{$\frac{1}{50}$}} (0,3)
		(v2) edge node [right] {\small{$\frac{2}{50}$}} (2.25,3)
		(v3) edge node [right] {\small{$\frac{3}{50}$}} (4.5,3)
		(v4) edge node [right] {\small{$\frac{m-1}{50}$}} (6.75,3)
		(v5) edge[est] node [right] {\small{$\frac{m}{50}$}} (v6)
		(0,3) edge (5,3)
		(6.25,3) edge[est] (v6);
		
	\end{tikzpicture}
	\vspace*{-0.2cm}
\caption{The task graph of the car wash problem.}\label{fig:carwash}
\end{figure}

To illustrate the model, we consider the car wash problem once more. 
Assume that the car has to be washed during the next $m$ days, where $m > 50$. 
The task graph $G$ is depicted in Figure~\ref{fig:carwash}. 
For each day $i$, with $1\leq i \leq m$, there is a vertex $v_i$. 
Let $v_1$ be the source. 
There is an edge $(v_i,t)$ of cost $i/50$ representing the action that Alice washes the car on day $i$. 
In order to keep the drawing simple, the edges $(v_i,t)$ merge in Figure~\ref{fig:carwash}.
Moreover, for every $i<m$ there is an edge $(v_i,v_{i+1})$ of cost~$0$ that represents the postponement of the job from day $i$ to the next day.
Assume for now that Alice is located at some $v_i$, with $i<m$.
Her perceived cost for procrastination is at least $\beta (i+1)/50$.
This lower bound is tight if Alice plans to traverse the edges $(v_i,v_{i+1})$ and $(v_{i+1},t)$.
Alternatively, her perceived cost for using $(v_i,t)$ and washing the car on day $i$ is $i/50$.
Remember that $\beta=1/3$.
It follows that $\zeta(v_i) = \beta (i+1)/50$, which means that Alice always prefers to wash the car on the next day instead of the current day.
Moreover, if $i<50$, then $\zeta(v_i) \leq \beta r$ for the reward of $r=1$ provided by the family at $t$. 
Thus Alice procrastinates and moves along $(v_i,v_{i+1})$. 
Note that her planning is time-inconsistent. 
On day $i$ she intends to follow the path $(v_i,v_{i+1}),(v_{i+1},t)$.
However when located at $v_{i+1}$ she pursues a different strategy. 
Once Alice reaches $v_{50}$, she realizes that $\zeta(v_{50})=\beta (50+1)/50$ exceeds the perceived value of the reward $\beta r$ and abandons. 
Thus $G$ is not motivating.

Next assume that we delete $(v_{16},v_{17})$ from $G$. 
In other words, we remove a procrastination edge and thereby set a deadline at day $i=16$. 
Let $G'$ be the resulting graph. 
When Alice reaches $v_{16}$ she can not procrastinate anymore and her perceived cost is $\zeta_{G'}(v_{16}) = 16/50$ which is less than the perceived value $\beta r =1/3$ of the reward. Hence Alice washes the car and reaches $t$. 
Subgraph $G'$ is motivating.
However, it is interesting to observe, that there is no reward configuration $r$ within a budget less than $(m/50)/\beta$ that is motivating in the original task graph $G$.
This is due to the fact that no matter how much reward is placed at $t$, Alice will always prefer to procrastinate until day $m$, when her cost for washing the car is $m/50$.

To illustrate the strengths of reward configurations, we consider a second scenario.
Suppose that at day $i=50$ Alice's family offers her a new opportunity to earn pocket money.
If she first washes the family's car, which now incurs a cost of $1$, and afterwards also cleans her room, which due to years of neglect incures a cost of $6$, she receives $10$ Euros.
Secretly, the family does not care about Alice cleaning her room.
They only try to trick her into washing the car for free.
We model this project with a new task graph $G$ that consists of a path from $s$ to $t$ via an intermediate vertex $v$ and another path from $v$ to $t$ via an intermediate vertex $w$.
The edge $(s,v)$ corresponds to the job of washing the car and has a cost of $1$, while $(v,w)$ is the job of cleaning Alice's room and has a cost of $6$.
The edges $(v,t)$ and $(w,t)$ are of cost~$0$.
Assuming that $\beta=1/3$, there is a reward configuration $r$ for which the family can motivate Alice to complete the project within a budget of $0$.
Setting $r(w)=10$, Alice traverses $(s,v)$ with a lowest perceived cost of $\zeta_r(s)=-1/3$.
This cost is realized along the edges $(s,v)$, $(v,w)$ and $(w,t)$.
When at $v$, Alice perceives cost of $8/3$ for traversing $(v,w)$ and cleaning her room but $0$ for ending the project right away along $(s,t)$.
Thus she changes her plan and moves to $t$ without collecting a reward.
Interestingly, there is no motivating subgraph of $G$ for a reward less than $3$ if the reward must be placed at $t$.
This suggests that, depending on the structure of the task graph, the performance of our two design strategies may vary drastically.


\section{The complexity of finding motivating subgraphs}\label{sec:comp}

In this section we first observe that if $\beta=0$ or $\beta=1$, then the problem of finding a motivating subgraph can be solved in polynomial-time.
We then prove that the decision problem, which we refer to as MOTIVATING SUBGRAPH (MS), is NP-complete for general $\beta\in(0,1)$. 
Our proof is based on a reduction from $k$~DISJOINT CONNECTING PATHS ($k$-DCP), cf.~\cite{GJ}.
Lynch~\cite{L} showed that $k$-DCP is NP-complete in undirected graphs.
In the Appendix, by adapting Lynch's proof, we show that $k$-DCP is also NP-complete in directed acyclic graphs.

\begin{proposition}\label{prop:01beta}
A motivating sub graph can be found in polynomial time if $\beta=0$ or $\beta=1$. 
\end{proposition}
\begin{proof}
We start with $\beta=0$. 
In this case a subgraph $G'$ of $G$ is only motivating if at every vertex of $G'$ the agent's perceived cost is~$0$. 
Hence $G$ contains a motivating subgraph if and only if $G$ contains a path from $s$ to $t$ such that all of its edges have cost~$0$. 
Any such path is a motivating subgraph. 
If $\beta=1$, then the agent follows a cheapest path from $s$ to $t$ in any subgraph. 
Hence $G$ contains a motivating subgraph if and only if there exists a path from $s$ to $t$ with a total edge cost of at most $r$.
Should such a path exist, then $G$ is its own motivating subgraph.
Clearly, a motivating subgraph can be found in polynomial-time for both cases $\beta=0$ and $\beta=1$.
\end{proof}

We now formally define the decision problem MS. 

\begin{definition}[MOTIVATING SUBGRAPH]
Given a task graph $G$, a reward $r$ and a bias factor $\beta\in[0,1]$, decide the existence of a motivating subgraph of $G$.
\end{definition}

The following proposition, while being interesting in its own right, implies that MS is contained in the complexity class NP.

\begin{proposition}\label{prop:2}
For any task graph $G$, reward $r$ and bias factor $\beta$, it can be decided in polynomial-time if $G$ is motivating.
\end{proposition}
\begin{proof}
We modify $G$ in the following way.
For each vertex $v$, we calculate the lowest perceived cost $\zeta_G(v)$.
Next, we take a copy of $G$, say $G'$, and remove all edges $(v,w)$ from $G'$ such that $\zeta_G(v) < c_G(v,w) + \beta d_G(w).$
In other words, we remove all edges from $G'$ that do not minimize the agents perceived cost.
Because the vertices that can be reached from $s$ in $G'$ are exactly those vertices that are visited by the agent in $G$, $G$ is motivating if and only if $\zeta_G(v) \leq \beta r$ for all vertices that can be reached from $s$ in $G'$.
The latter condition can be checked in polynomial-time by standard graph search algorithms.
\end{proof}

Before we prove NP-hardness of MS, we restate the definition of $k$-DCP as a brief reminder.

\begin{definition}[$k$ DISJOINT CONNECTING PATHS]
Given a directed acyclic graph $H$ and $k$ disjoint vertex pairs $(s_1,t_1), \ldots, (s_k,t_k)$, decide if $H$ contain $k$ mutually vertex-disjoint paths, one connecting every $s_i$ to the corresponding $t_i$.
\end{definition}

Furthermore, we want to introduce to a simple but useful lemma, which lets us set prices along a path of arbitrary length $k$, such that at every vertex, except for the last, the perceived cost of following the path to its end is exactly~$1$.
Such price structures will be a reoccurring feature of the reductions in Theorem~\ref{th:comp} and~\ref{th:approx}.

\begin{lemma}\label{lem:geoseries}
For every positive integer $k$ and bias factor $\beta \in [0,1]$ it holds that 
\[(1 - \beta)^k + \beta\biggl(\sum_{i=0}^{k-1}(1 - \beta)^i\biggr) = 1.\]
\end{lemma}
\begin{proof}
If $\beta$ is equal to $0$, this claim is easy to verify.
However, should $\beta$ be greater than $0$, the geometric series $\sum_{i=0}^{k-1}(1 - \beta)^i$ and can be rewritten as $(1-(1-\beta)^k)/\beta$, which in turn implies that
\[(1 - \beta)^k + \beta\biggl(\sum_{i=0}^{k-1}(1 - \beta)^i\biggr) = (1-\beta)^k + \beta\biggl(\frac{1-(1-k)^k}{\beta}\biggr) = 1.\]
\end{proof}

We are now ready to prove NP-completeness of MS.

\begin{theorem}\label{th:comp}
MS is NP-complete, for any bias factor $\beta\in(0,1)$. 
\end{theorem}
\begin{proof}
By Proposition~\ref{prop:2} we can take any motivating subgraph $G'$ as certificate for a "yes"-instance of MS.
Hence MS is in NP.
In the following we will present a polynomial-time reduction from $k$-DCP to show NP-hardness.
This establishes the theorem. 
Consider an instance ${\cal I}$ of $k$-DCP, consisting of a directed acyclic graph $H$ and $k$ disjoint vertex pairs $(s_1,t_1), \ldots, (s_k,t_k)$. 
We construct an instance ${\cal J}$ of MS that is composed of a task graph $G$, a bias factor $\beta$ and a reward $r$.
The graph $H$ will be embedded into $G$ in such a way that $G$ has a motivating subgraph if and only if H has $k$ disjoint connecting paths. 

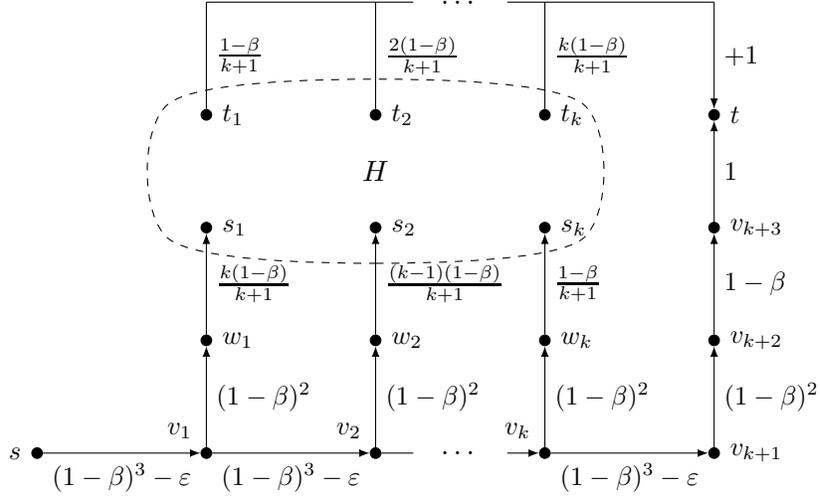
\begin{figure}[t]
	\center
		\begin{tikzpicture}[nst/.style={draw,circle,fill=black,minimum size=4pt,inner sep=0pt]}, est/.style={draw,>=latex,->}]
			
			\node[nst] (v1) at (0,0) [label=left:\small{$s$}] {};
			\node[nst] (v2) at (2.25,0) [label=above left:\small{$v_1$}] {};
			\node[nst] (v3) at (4.5,0) [label=above left:\small{$v_2$}] {};
			\node[nst] (v4) at (6.75,0) [label=above left:\small{$v_{k}$}] {};
			\node[nst] (v5) at (9,0) [label=right:\small{$v_{k+1}$}] {};
			\node[nst] (v6) at (9,1.5) [label=right:\small{$v_{k+2}$}] {};
			\node[nst] (v7) at (9,3) [label=right:\small{$v_{k+3}$}] {};
			\node[nst] (v8) at (9,4.5) [label=right:\small{$t$}] {};
			\node[nst] (v9) at (2.25,1.5) [label=right:\small{$w_1$}] {};
			\node[nst] (v10) at (4.5,1.5) [label=right:\small{$w_2$}] {};
			\node[nst] (v11) at (6.75,1.5) [label=right:\small{$w_k$}] {};
			\node[nst] (v12) at (2.25,3) [label=right:\small{$s_1$}] {};
			\node[nst] (v13) at (4.5,3) [label=right:\small{$s_2$}] {};
			\node[nst] (v14) at (6.75,3) [label=right:\small{$s_k$}] {};
			\node[nst] (v15) at (2.25,4.5) [label=right:\small{$t_1$}] {};
			\node[nst] (v16) at (4.5,4.5) [label=right:\small{$t_2$}] {};
			\node[nst] (v17) at (6.75,4.5) [label=right:\small{$t_k$}] {};
		
			\draw[dashed] plot [smooth cycle] coordinates {(1.75,2.85) (4.5,2.525) (7.25,2.85) (7.25,4.65) (4.5,4.975) (1.75,4.65)};	
			\node at (4.5,3.75) {$H$};
			\node at (5.625,0) {$\dots$};
			\node at (5.625,6) {$\dots$};
		
			\path (v1) edge[est] node [below] {\small{$(1-\beta)^3 - \varepsilon$}} (v2)
			(v2) edge[est] node [below] {\small{$(1-\beta)^3 - \varepsilon$}} (v3)
			(v3) edge (5,0)
			(6.25,0) edge[est] (v4)
			(v4) edge[est] node [below] {\small{$(1-\beta)^3 - \varepsilon$}} (v5)
			(v5) edge[est] node [right] {\small{$(1-\beta)^2$}} (v6)
			(v6) edge[est] node [right] {\small{$1-\beta$}} (v7)
			(v7) edge[est] node [right] {\small{$1$}} (v8)
			(v2) edge[est] node [right] {\small{$(1-\beta)^2$}} (v9)		
			(v3) edge[est] node [right] {\small{$(1-\beta)^2$}} (v10)
			(v4) edge[est] node [right] {\small{$(1-\beta)^2$}} (v11)
			(v9) edge[est] node [right] {\small{$\frac{k(1-\beta)}{k+1}$}} (v12)		
			(v10) edge[est] node [right] {\small{$\frac{(k-1)(1-\beta)}{k+1}$}} (v13)
			(v11) edge[est] node [right] {\small{$\frac{1-\beta}{k+1}$}} (v14)
			(v17) edge node [right] {\small{$\frac{k(1-\beta)}{k+1}$}} (6.75,6)		
			(v16) edge node [right] {\small{$\frac{2(1-\beta)}{k+1}$}} (4.5,6)
			(v15) edge node [right] {\small{$\frac{1-\beta}{k+1}$}} (2.25,6)
			(2.25,6) edge  (5,6)
			(6.25,6) edge  (9,6)
			(9,6) edge[est] node [right] {\small{$+ 1$}} (v8);
			
		\end{tikzpicture}
	\caption{The graph $G$ with an embedding of $H$.}\label{fig:g1}
\end{figure}

We proceed to describe the MS instance ${\cal J}$.
Let $\beta\in(0,1)$ be any value with the property that its encoding length is polynomial in that of ${\cal I}$.
Set $r= 1/\beta$. 
The task graph $G$ is constructed as follows, see also Figure~\ref{fig:g1}.
It consists of a source $s$ and a target $t$.
These two vertices are connected by a directed {\em main path\/} along intermediate vertices $v_1, \ldots, v_{k+3}$.
The first $k+1$ edges of this main path each have a cost of $(1-\beta)^3-\varepsilon$, where $\varepsilon$ is a positive constant satisfying
\[\varepsilon < \min\biggl\{\beta  \frac{1-\beta}{k+1}, \beta\frac{(1-\beta)^3}{1+\beta}\biggr\}.\]
The last three edges of the main path, connecting $v_{k+1}$ to $t$, have a cost of $(1-\beta)^2$, $1-\beta$ and $1$ respectively. 
Additionally, $G$ contains $k$ {\em shortcuts\/} that connect every $v_i$, with $1\leq i\leq k$, to $t$ via an embedding of $H$.
More formally, $H$ is added to $G$. 
The $i$-th shortcut starts at $v_i$.
It visits a distinct vertex $w_i$ along an edge of cost $(1-\beta)^2$.
Vertex $w_i$ is connected to $s_i$ in $H$.
The edge cost is $(k+1-i)(1-\beta)/(k+1)$.
Finally, vertex $t_i$ of $H$ is connected to $t$ via an edge of cost $i(1-\beta)/(k+1)+1$. 
In Figure~\ref{fig:g1} the latter edge cost is shown as two terms, namely \ $i(1-\beta)/(k+1)$ and $1$, in order to keep the labels of the parallel edges $(t_i,t)$ simple. 
Note that for any shortcut $i$, the edge costs of $(w_i,s_i)$ and $(t_i,t)$ complement each other, i.e.\ they sum to exactly $(1-\beta) +1$.
The edges of $H$ all have a cost of~$0$.
We remark that at every vertex different from $t$, the perceived value of the reward is $\beta r = 1$. 
The resulting graph $G$ is acyclic and its encoding length is polynomial in that of ${\cal I}$.
We next prove that ${\cal I\/}$ has a solution if and only if ${\cal J}$ has one. 

($\Longrightarrow$) First assume that ${\cal I}$ has a solution, i.e.\ there exist $k$ vertex-disjoint paths, one connecting every $s_i$ to the corresponding $t_i$. 
In the embedding of $H$ we remove all edges, except for the $k$ vertex-disjoint paths.
Let $G'$ be the resulting subgraph of $G$.
We will show that $G'$ is motivating, for $r=1/\beta$.
More specifically, we will argue that the agent travels along the main path from $s$ to $t$.
If the agent resides at one of the first $k$ vertices $v_i$ it has two options.
Either it traverses $(v_{i}, v_{i+1})$ and follows the main path, or it takes $(v_{i}, w_{i})$ and walks along the $i$-th shortcut.
Let $s=0$.
For $0\leq i<k$, the perceived cost of traversing $(v_{i}, v_{i+1})$ and following the $(i+1)$-st shortcut afterwards is $(1-\beta)^3 - \varepsilon + \beta((1-\beta)^2 + (1-\beta) + 1)$.
According to Lemma~\ref{lem:geoseries}, the value of the perceived cost simplifies to $1-\varepsilon$.
Note that similar calculations are scattered throughout the entire proof.
For the sake of brevity, Lemma~\ref{lem:geoseries} will not be referred to explicitly each time.
If the agent is at $v_k$, its perceived cost in following the main path to $t$ is also $1-\varepsilon$.
Hence, taking $(v_{i}, v_{i+1})$ is a motivating option.
In contrast, if the agent resides at $v_i$, with $1\leq i \leq k$, and plans to traverse $(v_{i},w_{i})$, following the $i$-th shortcut, its perceived cost is $1$.
Although this option is also motivating, it is perceived as more expensive than taking $(v_{i}, v_{i+1})$.
As a result, the agent follows the main path until it reaches $v_{k+1}$.
At this point the agent has no option but to stay on the main path.
The perceived cost at any of the vertices $v_{k+1}$, $v_{k+2}$ and $v_{k+3}$ is $1$. 
Thus subgraph $G'$ is indeed motivating. 

($\Longleftarrow$) Next assume that ${\cal I}$ does not have a solution. 
We prove that no subgraph $G'$ of $G$ is motivating.
Consider any subgraph $G'$.
Observe that $G'$ is only motivating if the agent never leaves the main path.
Otherwise the agent must visit some $t_i$ on its way to $t$ at which point it perceives a cost of $i(1-\beta)/(k+1) + 1 >1$ and abandons. 
We therefore focus on subgraphs $G'$ that contain all edges of the main path.
More specifically, we focus on subgraphs $G'$ in which the agent walks along the main path.
We say that the $i$-th shortcut is {\em degenerate\/} if the total edge cost of a cheapest path from $v_i$ to $t$ via $(v_i,w_i)$ is different from the target value $\theta=(1-\beta)^2 + (1-\beta) +1$.
In particular, the $i$-th shortcut is degenerate if there is no path from $v_i$ to $t$ via $(v_i,w_i)$, in which case the perceived cost of the shortcut is infinite. 
Note that by construction, every degenerate shortcut must miss the target value by $(1-\beta)/(k+1)$ or more.

We first argue that there is at least one degenerate shortcut in $G'$.
For the sake of contradiction, assume no such shortcut exists.
This means that there is a cheapest path $P_i$ from $v_i$ to $t$ via $(v_i,w_i)$ for all $1\leq i\leq k$.
By construction, $P_i$ traverses $(w_i,s_i)$.
Remember that the total cost of $P_i$ must sum up to $\theta$.
The only way to achieve this is if $P_i$ ends in $(t_i,t)$.
Furthermore, $P_i$ must be vertex disjoint from all other paths $P_j$ with $j<i$.
Otherwise $P_i$ would not be a shortest path from $v_i$ to $t$, given that $c(t_j,t) < c(t_i,t)$.
However, this implies that there are $k$ vertex disjoint paths in $H$, one from each $s_i$ to the corresponding $t_i$, which contradicts the assumption that ${\cal I}$ has no solution.

Now that we established the existence of a degenerate shortcut, we distinguish two cases. 
Either there exists a degenerate shortcut $i$ such that the cost of a cheapest path from $v_i$ to $t$ via $(v_i,w_i)$ is less than $\theta$ or for each degenerate shortcut $i$ the cost of a cheapest path from $v_i$ to $t$ via $(v_i,w_i)$ is greater than $\theta$. 

We study the first case first.
Let $i$ be the largest index of a degenerate shortcut such that the cheapest path from $v_i$ to $t$ via $(v_i,w_i)$ is less than $\theta$.
When located at $v_i$ the agent perceives cost less or equal to
\[(1-\beta)^2 + \beta\biggl((1-\beta)+1-\frac{1-\beta}{k+1}\biggr)=1-\beta\frac{1-\beta}{k+1}<1-\varepsilon\]
along $(v_i,w_i)$.
Conversely, in planning a cheapest path along $(v_i,v_{i+1})$ and following a subsequent shortcut or the main path, the agent perceives a cost of at least $1-\varepsilon$.
This holds true because all subsequent shortcuts are of cost $\theta$ or more.
By choice of $\varepsilon$, the perceived cost along $(v_i,w_i)$ is less than the perceived cost along $(v_i,v_{i+1})$.
However, this contradicts our assumption that the agent stays on the main path. 

We finally study the second case.
Suppose that the $i$-th shortcut is degenerate and consider the agent planing its path from $v_{i-1}$ to $t$ via $(v_{i-1},v_i)$.
The agent has two options.
If the agent plans to follow the $i$-th shortcut, it perceives a cost greater or equal to
\[(1-\beta)^3 -\varepsilon + \beta\biggl((1-\beta)^2+(1-\beta) +1 + \frac{1-\beta}{k+1}\biggr) = 1+\beta \frac{1-\beta}{k+1}-\varepsilon > 1.\]
The inequality holds by choice of $\varepsilon$. 
If the agent plans tor traverse $(v_{i-1},v_i)$ instead, taking either a shortcut $j>i$ or following the main path all the way to $t$, it perceives a cost of at least 
\[(1-\beta)^3 -\varepsilon + \beta\big((1-\beta)^3 -\varepsilon + (1-\beta)^2+(1-\beta) +1\big) = 1 + (1+\beta)\biggl(\beta\frac{(1-\beta)^3}{1+\beta} -\varepsilon\biggr)>1.\]
This holds true because no shortcut is of cost less than $\theta$.
Once more, the perceived cost is greater than $1$ by definition of $\varepsilon$. 
Hence the agent certainly abandons at $v_{i-1}$, which shows that $G'$ cannot be motivating.
\end{proof}


\section{Approximating optimum rewards}\label{sec:approx}

Considering that the decision problem MS is NP-hard, the next and arguably natural question is whether there exist good approximation algorithms. 
Hence we formulate MS as an optimization problem.

\begin{definition}[MOTIVATING SUBGRAPH OPT]
Given a task graph $G$ and a bias factor $\beta\in(0,1)$, determine the minimum reward $r$ to place at $t$ such that $G$ contains a motivating subgraph.
\end{definition}

We present two simple algorithms.
The first algorithm is designed for small values of $\beta$.
The second algorithm computes good solutions for large $\beta$. 
The algorithms output $r$ as well as a corresponding motivating subgraph $G'$.
Both strategies are somewhat reminiscent of Proposition~\ref{prop:01beta}.
A combination of them yields a $(1+\sqrt{n}$)-approximation, for any $\beta\in(0,1)$. 

Suppose that $\beta$ is small.
Then the agent is highly oblivious to the future. 
Consequently it is sensible to let the agent travel along a path that minimizes the maximum cost of any edge. 
We call a path with this property a {\em minmax path\/}. 
A minmax path can be computed easily in polynomial-time.
For instance, starting with an empty subgraph, insert the edges of $G$ in non-decreasing order of cost until $s$ and $t$ become connected for the first time.
Next, choose one of the possibly several paths that connect $s$ and $t$ in the subgraph as minmax path.
Our first algorithm, called {\sc MinmaxPathApprox}, computes a minmax path $P$ and returns the corresponding $G'$ containing only the edges of $P$.
Furthermore, the algorithm sets $r$ according to the maximum over all perceived cost along $P$, or more formally $r = \max\{\zeta_{G'}(v) \mid v \in P\}/\beta$.
Clearly, this reward is sufficient to make $G'$ motivating.

\begin{proposition}\label{prop:alg1}
{\sc MinmaxPathApprox} achieves an approximation ratio of $1+\beta n$, for any $\beta\in(0,1)$. 
\end{proposition}
\begin{proof}
Let $c$ denote the maximum cost of any edge along the path $P$ computed by {\sc MinmaxPathApprox}.
By definition of $P$ the agent must encounter an edge of cost at least $c$ in any motivating subgraph.
Thus the optimum reward is lower bounded by $c/\beta$.
Conversely, the cost of every edge in $P$, of which there are at most $n-1$, is upper bounded by $c$.
This means that {\sc MinmaxPathApprox} returns a reward $r$, which is upper bounded by
\[r = \frac{\max\{\zeta_{G'}(v) \mid v \in P\}}{\beta} \leq \frac{c}{\beta} + (n-2)c \leq \frac{c}{\beta} + nc,\]
which yields the desired approximation ratio of $1+\beta n$. 
\end{proof}

Our second algorithm, called {\sc CheapestPathApprox}, computes a path $P$ of minimum total cost from $s$ to $t$ and returns the corresponding $G'$ containing only the edges of $P$.
Again, the algorithm sets the reward to $r = \max\{\zeta_{G'}(v) \mid v \in P\}/\beta$.

\begin{proposition}\label{prop:alg2}
{\sc CheapestPathApprox} achieves an approximation ratio of $1/\beta$, for any $\beta\in(0,1)$. 
\end{proposition}
\begin{proof}
Let $P$ be the path computed by {\sc CheapestPathApprox}.
At any vertex $v$ of $P$, the agent perceives a cost of at most $d_{G'}(v)$, which is upper bounded by $d_{G'}(s)$.
Thus $d_{G'}(s)/\beta$ is an upper bound on the the reward $r$ calculated by {\sc CheapestPathApprox}.
In an optimal solution, when located at $s$, the agent is faced with a cost of at least $\beta d_{G}(s)$.
Consequently, a reward of at least $d_{G}(s)$ is required to motivate the agent.
Because $P$ is a cheapest path, it holds that $d_{G}(s) = d_{G'}(s)$, which establishes an approximation ratio of $1/\beta$.
\end{proof}

Let {\sc CombinedApprox} be the combined algorithm that chooses {\sc MinmaxPathApprox} if $\beta \leq 1/\sqrt{n}$ and {\sc CheapestPathApprox} otherwise. 
Propositions~\ref{prop:alg1} and~\ref{prop:alg2} immediately imply the following result. 
\begin{theorem}\label{th:alg3}
{\sc CombinedApprox} achieves an approximation ratio of $1+\sqrt{n}$, for any $\beta\in(0,1)$. 
\end{theorem}

We next prove that, although our $(1+\sqrt{n})$-approximation algorithm is simple, it achieves the best possible performance guarantee, up to a small constant factor, that can be hoped for in polynomial-time.
For the proof of the theorem we need the next technical lemma.

\begin{lemma}\label{lem:rho}
For any integer $\rho$, with $\rho\geq1$, it holds that 
\[\biggl(1-\frac{1}{3\rho+3}\biggr)^{3\rho+3} > \frac{1}{3}.\]
\end{lemma}
\begin{proof}
The sequence $(1-1/n)^n$ is monotonically increasing for $n\geq 1$.
Hence it holds that
\[\biggl(1-\frac{1}{3\rho+3}\biggr)^{3\rho+3} \geq (1-{1/6})^6 > {1/3}.\]
\end{proof}

\begin{theorem}\label{th:approx}
MS-OPT is NP-hard to approximated within a ratio of $1/4 \sqrt{n}$.
\end{theorem}
\begin{proof}
Again we present a reduction from $k$-DCP. 
Let $\rho$ be an arbitrary positive integer.
The best choice of $\rho$ will be determined later.
Given an instance ${\cal I}$ of $k$-DCP we construct an instance ${\cal J}$ of MS-OPT, consisting of a task graph $G$ and a bias factor $\beta$, having the following two properties. 
(a)~If ${\cal I}$ has a solution, then $G$ has a subgraph that is motivating for a reward of $r=1/\beta$.
(b)~If ${\cal I}$ does not have a solution, then no subgraph of $G$ is motivating for a reward of at most $r=\rho /\beta$.
Hence any algorithm that achieves an approximation ratio of $\rho$ or better must solve instances of $k$-DCP. 

We begin with the description of ${\cal J}$.
As before an instance of ${\cal I}$ is specified by a graph $H$ together with $k$ vertex pairs $(s_1,t_1), \ldots, (s_k,t_k)$. 
Considering that Proposition~\ref{prop:alg2} gives a $(1/\beta)$-approximation, the bias factor of ${\cal J}$ cannot be chosen arbitrarily anymore.
It must be less than $1/\rho$.
For convenience we set $\beta=1/(3\rho+3)$.
From a structural point of view $G$ of ${\cal J}$ consists of two units, a {\em central unit\/} and an {\em amplification unit\/}.
The central unit contains an embedding of $H$.
The amplification unit precedes the central unit and increases any approximation error that might occur in the central unit. 

\begin{figure}[t]
\center
\begin{tikzpicture}[nst/.style={draw,circle,fill=black,minimum size=4pt,inner sep=0pt]}, est/.style={draw,>=latex,->}]
			
			\node[nst] (v1) at (-0.3,0) [label=left:\small{$u_{9\rho^2}$}] {};
			\node[nst] (v2) at (2.25,0) [label=above left:\small{$v_1$}] {};
			\node[nst] (v3) at (4.5,0) [label=above left:\small{$v_2$}] {};
			\node[nst] (v4) at (6.75,0) [label=above left:\small{$v_{k}$}] {};
			\node[nst] (v5) at (9,0) [label=right:\small{$v_{k+1}$}] {};
			\node[nst] (v6) at (9,1.5) [label=right:\small{$v_{k+2}$}] {};
			\node[nst] (v7) at (9,3) [label=right:\small{$v_{k+3\rho+3}$}] {};
			\node[nst] (v8) at (9,4.5) [label=right:\small{$t$}] {};
			\node[nst] (v9) at (2.25,1.5) [label=right:\small{$w_1$}] {};
			\node[nst] (v10) at (4.5,1.5) [label=right:\small{$w_2$}] {};
			\node[nst] (v11) at (6.75,1.5) [label=right:\small{$w_k$}] {};
			\node[nst] (v12) at (2.25,3) [label=right:\small{$s_1$}] {};
			\node[nst] (v13) at (4.5,3) [label=right:\small{$s_2$}] {};
			\node[nst] (v14) at (6.75,3) [label=right:\small{$s_k$}] {};
			\node[nst] (v15) at (2.25,4.5) [label=right:\small{$t_1$}] {};
			\node[nst] (v16) at (4.5,4.5) [label=right:\small{$t_2$}] {};
			\node[nst] (v17) at (6.75,4.5) [label=right:\small{$t_k$}] {};
		
			\draw[dashed] plot [smooth cycle] coordinates {(1.75,2.85) (4.5,2.525) (7.25,2.85) (7.25,4.65) (4.5,4.975) (1.75,4.65)};	
			\node at (4.5,3.75) {$H$};
			\node at (5.625,0) {$\dots$};
			\node at (5.625,6) {$\dots$};
			\node at (9,2.35) {$\vdots$};
			
			\path (v1) edge[est] node [below] {\small{$(1-\beta)^{3\rho+3}-\varepsilon$}} (v2)
			(v2) edge[est] node [below] {\small{$(1-\beta)^{3\rho+3}-\varepsilon$}} (v3)
			(v3) edge (5,0)
			(6.25,0) edge[est] (v4)
			(v4) edge[est] node [below] {\small{$(1-\beta)^{3\rho+3}-\varepsilon$}} (v5)
			(v5) edge[est] node [right] {\small{$(1-\beta)^{3\rho+2}$}} (v6)
			(v6) edge (9,1.875)
			(9,2.625) edge[est] (v7)
			(v7) edge[est] node [right] {\small{$1$}} (v8)
			(v2) edge[est] node [right] {\small{$(1-\beta)^{3\rho+2}$}} (v9)		
			(v3) edge[est] node [right] {\small{$(1-\beta)^{3\rho+2}$}} (v10)
			(v4) edge[est] node [right] {\small{$(1-\beta)^{3\rho+2}$}} (v11)
			(v9) edge[est] node [right] {\small{$\frac{k(1 - \beta)^{3\rho+1}}{k+1}$}} (v12)		
			(v10) edge[est] node [right] {\small{$\frac{(k-1)(1-\beta)^{3\rho+1}}{k+1}$}} (v13)
			(v11) edge[est] node [right] {\small{$\frac{(1-\beta)^{3\rho+1}}{k+1}$}} (v14)
			(v17) edge node [right] {\small{$\frac{k(1-\beta)^{3\rho+1}}{k+1}$}} (6.75,6)		
			(v16) edge node [right] {\small{$\frac{2(1-\beta)^{3\rho+1}}{k+1}$}} (4.5,6)
			(v15) edge node [right] {\small{$\frac{(1-\beta)^{3\rho+1}}{k+1}$}} (2.25,6)
			(2.25,6) edge  (5,6)
			(6.25,6) edge  (9,6)
			(9,6) edge[est] node [right] {\small{$+\sum\limits_{j=0}^{3\rho}(1-\beta)^j$}} (v8);
			
\end{tikzpicture}
\caption{The central unit of $G$.}\label{fig:g2}
\end{figure}

The central unit, depicted in Figure~\ref{fig:g2}, has the same overall structure as the graph constructed in the proof of Theorem~\ref{th:comp}.
There exists a {\em main path} and $k$ {\em shortcuts}, linking to an embedding of $H$.
However, there are important differences.
The main path is longer and edge costs are different. 
More specifically the main path starts at a vertex $u_{9\rho^2}$, which is the last vertex of the amplification unit, and ends at vertex $t$.
The path consists of $k+3\rho+3$ intermediate vertices $v_1, \ldots, v_{k+3\rho+3}$.
The first $k+1$ edges of the main path each have a cost of $(1-\beta)^{3\rho+3} -\varepsilon$, where $\varepsilon$ is a positive value satisfying
\[\varepsilon < \min\biggl\{\beta\frac{(1-\beta)^{3\rho+1}}{k+1}, \beta\frac{(1-\beta)^{3\rho+3}}{1+\beta}, \frac{1}{1+\rho}, (1-\beta)^{3\rho+3}-\frac{1}{3}\biggr\}.\] 
Note that $(1-\beta)^{3\rho+3}-1/3$ is a positive quantity according to Lemma~\ref{lem:rho}.
The remaining edges of the main path have increasing cost.
In particular, we set the cost of $(v_i,v_{i+1})$, with $k < i \leq k+3\rho+3$, to $(1-\beta)^{k+3\rho+3-i}$.
Each of the first $k$ vertices $v_i$ is the starting point of a shortcut.
Similar to the construction in Theorem~\ref{th:comp}, the $i$-th shortcut is routed through a distinct vertex $w_i$ before it enters the emending of $H$ and eventually reaches $t$.
The edges $(v_i,w_i)$ have a cost of $(1-\beta)^{3\rho+2}$.
Furthermore, the edges $(w_i,s_i)$ and $(t_i,t)$ are of cost $(k+1-i)(1-\beta)^{3\rho+1}/(k+1)$ and $i(1-\beta)^{3\rho+1}/(k+1)+ \sum_{j=0}^{3\rho} (1-\beta)^j$ respectively.
In Figure~\ref{fig:g2}, the later edge cost is shown as two terms, namely $i(1-\beta)^{3\rho+1}/(k+1)$ and $\sum_{j=0}^{3\rho} (1-\beta)^j$, in order to keep the labels of the parallel edges $(t_i,t)$ simple.
Note that for $1 \leq i \leq k$, the costs of $(v_i,s_i)$ and $(t_i,t)$ sum to exactly $\sum_{j=0}^{3\rho+1} (1-\beta)^j$.
All edges of $H$ have cost~0.

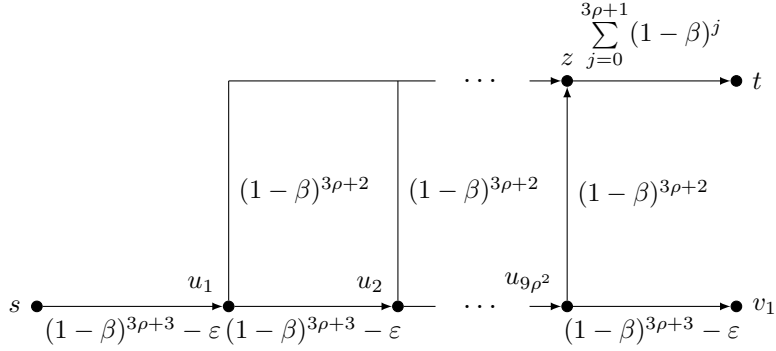
\begin{figure}[t]
\center
		\begin{tikzpicture}[nst/.style={draw,circle,fill=black,minimum size=4pt,inner sep=0pt]}, est/.style={draw,>=latex,->}]
			
			\node[nst] (v1) at (-0.3,0) [label=left:\small{$s$}] {};
			\node[nst] (v2) at (2.25,0) [label=above left:\small{$u_1$}] {};
			\node[nst] (v3) at (4.5,0) [label=above left:\small{$u_2$}] {};
			\node[nst] (v4) at (6.75,0) [label=above left:\small{$u_{9\rho^2}$}] {};
			\node[nst] (v5) at (6.75,3) [label=above:\small{$z$}] {};
			\node[nst] (v6) at (9,3) [label=right:\small{$t$}] {};
			\node[nst] (v7) at (9,0) [label=right:\small{$v_1$}] {};
			
			\node at (5.625,0) {$\dots$};
			\node at (5.625,3) {$\dots$};	
		
			\path (v1) edge[est] node [below] {\small{$(1-\beta)^{3\rho+3}-\varepsilon$}} (v2)
			(v2) edge[est] node [below] {\small{$(1-\beta)^{3\rho+3}-\varepsilon$}} (v3)
			(v3) edge (5,0)
			(6.25,0) edge[est] (v4)
			(v2) edge node [right] {\small{$(1-\beta)^{3\rho+2}$}} (2.25,3)
			(v3) edge node [right] {\small{$(1-\beta)^{3\rho+2}$}} (4.5,3)
			(v4) edge[est] node [right] {\small{$(1-\beta)^{3\rho+2}$}} (v5)
			(2.25,3) edge (5,3)
			(6.25,3) edge[est] (v5)
			(v5) edge[est] node [above] {\small{$\sum\limits_{j=0}^{3\rho+1}(1-\beta)^j$}} (v6)
			(v4) edge[est] node [below] {\small{$(1-\beta)^{3\rho+3}-\varepsilon$}} (v7);
			
		\end{tikzpicture}
\caption{The amplification unit of $G$.}\label{fig:g3}
\end{figure}

Next we describe the amplification unit, which is shown in Figure~\ref{fig:g3}. 
Starting at vertex $s$, there is a directed path to $u_{9\rho^2}$, called the {\em amplification path}, that consists of intermediate vertices $u_1, \ldots, u_{9\rho^2-1}$.
Each edge of the amplification path has a cost of $(1-\beta)^{3\rho+3}-\varepsilon$. 
From every $u_i$, there is also an edge to a vertex $z$ of cost $(1-\beta)^{3\rho+2}$.
Vertex $z$ is connected to $t$ via an edge of cost $\sum_{j=0}^{3\rho+1} (1-\beta)^j$.
In the following we prove the statements given in (a) and (b) above.

(a)~Assume that ${\cal I}$ has a solution. 
Let $G'$ be the subgraph of $G$ obtained by deleting all edges from the embedding of $H$ that do not belong to one of the $k$ vertex-disjoint paths in a fixed solution of ${\cal I}$. 
We claim that $G'$ is motivating with reward $r=1/\beta$.
Remember that the agent perceives $r$ as $1$ on all vertices of $G$ except for $t$.
Furthermore, we will use Lemma~\ref{lem:geoseries} to calculate the agent's perceived cost if not stated explicitly otherwise.
Let $u_0=s$. 
At every vertex $u_i$, with $i < 9\rho^2$, the agent's perceived cost for moving to $u_{i+1}$ and then directly to $t$, hence traversing edges $(u_i,u_{i+1})$, $(u_{i+1},z)$ and $(z,t)$, is $1-\varepsilon$.
Conversely, when residing at $u_i$, with $i \leq 9\rho^2$, the agent's perceived cost in moving directly to $t$ using edges $(u_{i+1},z)$ and $(z,t)$ is $1$.
Thus the agent moves along the edges of the amplification path, until it reaches $u_{9\rho^2}$.
At $u_{9\rho^2}$ the agent moves on to $u_1$ because its perceived cost in traversing $(u_{9\rho^2},v_1)$ and then following the first shortcut, which starts at $v_1$, is $1-\varepsilon$.
Similarly, when located at $v_i$, with $1\leq i<k$, agent perceives cost of $1-\varepsilon$ for traversing $(v_i,v_{i+1})$ and then taking the $(i+1)$-st shortcut.
In contrast, planning a path along $(v_i,w_i)$, with $1\leq i\leq k$, has a cost of $1$.
Thus the agent follows the main path until reaching $v_{k}$. 
By the same calculations, the agent moves to $v_{k+1}$. 
At this point the agent has no other option but to follow the main path.
Because the agent's perceived cost is $1$ for all vertices $v_{i}$, with $k<1\leq k+3\rho+3$, it eventually reaches $t$.

(b)~Suppose that ${\cal I}$ does not have a solution and consider any subgraph $G'$ of $G$. 
We first argue that if the agent leaves the amplification path or main path, then it abandons given a reward of at most $\rho/\beta$, which is perceived as $\rho$ at every vertex different from $t$.
If the agent leaves at some $u_i$, it must pass $z$ where it perceives cost of $\sum_{j=0}^{3\rho+1} (1-\beta)^j$.
However, it holds that
\[\sum_{j=0}^{3\rho+1} (1-\beta)^j = \sum_{j=0}^{3\rho+1} \biggl(1 - \frac{1}{3\rho+3}\biggr)^j > \sum_{j=0}^{3\rho+1} \biggl(1 - \frac{1}{3\rho+3}\biggr)^{3\rho+3} > \sum_{j=0}^{3\rho+1} \frac{1}{3} > \rho,\]
see Lemma~\ref{lem:rho} for the second inequality, and therefore the agent abandons. 
Similarly, if the agent leaves at some $v_i$, with $1 \leq i \leq k$, then it must pass one of the vertices $t_j$, where the perceived cost is greater than $\sum_{j=0}^{3\rho} (1-\beta)^j$ and consequently also greater than $\rho$.
Hence, we will restrict ourselves of subgraphs $G'$ in which the amplification path and main path are intact and assume that the agent walks them.

We say that the $i$-th shortcut is {\em degenerate\/} if the cost of a cheapest path from $v_i$ to $t$ via $(v_i,w_i)$ is different from the target value $\theta = \sum_{j=0}^{3\rho+2} (1-\beta)^j$.
In particular, the $i$-th shortcut is degenerate if there is no such path.
Note that by construction, every degenerate shortcut must miss the target value by $(1-\beta)^{3\rho+1}/(k+1)$ or more.
As in the proof of Theorem~\ref{th:comp}, the assumption that ${\cal I}$ has no solution implies the existence of a degenerate shortcut.
By the same argument given in Theorem~\ref{th:comp}, it is also clear that if there is a degenerate shortcut of cost less than $\theta$, the agent leaves the main path and abandons.
In the remainder of the analysis of (b) we therefore assume that for all shortcuts $i$, the cost of a cheapest path from $v_i$ to $t$ via $(v_i, w_i)$ is greater than $\theta$. 
We distinguish two cases depending on whether or not the first shortcut is degenerate. 

If the first shortcut is not degenerate, then there exists an integer $i$, with $1<i\leq k$, such that the $(i-1)$-st shortcut is not degenerate but shortcut $i$ is.
When the agent resides at $v_{i-1}$ and plans a cheapest path along $(v_{i-1}, w_{i-1})$, it perceives a cost of $1$. 
In contrast, traversing $(v_{i-1},v_i)$ and taking the next shortcut $i$, has a perceives a cost of at least 
\[(1-\beta)^{3\rho+3} -\varepsilon + \beta\biggl(\sum_{j=0}^{3\rho+2} (1-\beta)^j + \frac{(1-\beta)^{3\rho+1}}{k+1}\biggr)= 1 + \beta\frac{(1-\beta)^{3\rho+1}}{k+1}-\varepsilon >1.\]
The inequality holds by the choice of $\varepsilon$.
Moreover, there are no degenerate shortcuts of cost less than $\theta$.
Thus traversing $(v_{i-1},v_i)$ and continuing further on the main path, possibly taking a subsequent shortcut, the agent's perceived cost is at least 
\[(1-\beta)^{3\rho+3} -\varepsilon + \beta\biggl((1-\beta)^{3\rho+3} -\varepsilon + \sum_{j=0}^{3\rho+2} (1-\beta)^j\biggr) = 1+ (1+\beta)\biggl(\beta\frac{(1-\beta)^{3\rho+3}}{1 + \beta}- \varepsilon\biggr)>1.\]
Again, the inequality holds by choice of $\varepsilon$. 
Thus the agent leaves the main path at $v_{i-1}$ and abandons.

Finally, we study the case that the first shortcut is degenerate with cost of a cheapest path from $v_1$ to $t$ via $(v_1,w_1)$ greater than $\theta$.
Let $i$ be highest index of a vertex on the amplification path such that $u_i$ is connected to $t$ via $(u_i,z)$ and $(z,t)$ in $G'$.
The perceived cost of such a path is $1$.
Conversely, any path along $(u_i,u_{i+1})$, or $(u_{9\rho^2},v_1)$ assuming $i= 9\rho^2$, has a perceived cost greater than $1$ as calculated in the last paragraph. 
Thus the agent leaves the amplification path and abandons. 
However, if no $u_i$ is connected to $t$ via $(u_i,z)$ and $(z,t)$, then the agent's lowest perceived cost at $s$ is lower bounded by
\[(1-\beta)^{3\rho+3}-\varepsilon +\beta\biggl(9\rho^2\bigl((1-\beta)^{3\rho+3}-\varepsilon\bigr)+\sum_{j=0}^{3\rho+2} (1-\beta)^j\biggr) = 1 - \varepsilon + 9\beta\rho^2\bigl((1-\beta)^{3\rho+3}-\varepsilon\bigr)\]
Taking into account that $\beta = 1/(3\rho+3)$, we can further simplify this term to
\[1 - \varepsilon + 9\rho^2\frac{1/3 + ((1-\beta)^{3\rho+3} - 1/3 -\varepsilon)}{3\rho+3} > 1 - \varepsilon + 9\rho^2\frac{1/3}{3\rho+3} = \rho + \biggl(\frac{1}{1+\rho} - \varepsilon \biggr) > \rho.\]
Once more, the two inequalities hold by choice of $\varepsilon$.
Hence $G'$ is not motivating for a reward of at most $\rho/\beta$.

In order to finish the proof of the theorem we have to determine $\rho$. 
We set $\rho = m$, where $m$ is the number of vertices in $H$. 
The total number of vertices in $G$ is $n=2 + (9m^2+1)+(m+2k+3m+3)$. 
The first term accounts for $s$ and $t$, the following bracket accounts for the number of vertices in the amplification unit and the last bracket accounts for the number of vertices in the central unit.
Thus we have presented a polynomial-time reduction.
Moreover, it holds that $n\leq 9m^2+6m+6 < 16m^2$ for every $m\geq2$, which means that $\rho$ is lower bounded by $1/4 \sqrt{n}$.
\end{proof}


\section{Motivation through intermediate rewards}\label{sec:ir}
In this section, we generalize the original model of Kleinberg and Oren~\cite{KO} to incorporate the placement of rewards on arbitrary vertices instead of just $t$. 
The goal is to minimize the total value of the rewards collected by the agent as it travels from $s$ to $t$. 
The problem of finding an optimal reward configuration can be solved in polynomial-time if $\beta=0$ or $\beta =1$.
We prove that, for general $\beta\in(0,1)$, the corresponding decision problem, which we call MOTIVATING REWARD CONFIGURATION (MRC), is NP-complete. 
Furthermore, we show that the optimization version MRC-OPT, is NP-hard to approximate within any ratio greater or equal to $1$.

\begin{proposition}\label{prop:mrc}
An optimal reward configuration can be found in polynomial time for $\beta=0$ or $\beta=1$. 
\end{proposition}
\begin{proof}
First suppose that $\beta=0$.
In this case the agent does not perceive rewards placed at any vertex of $G$ and will only traverse edges of cost~$0$. 
Consider the subset $V'$ of the vertices that can be reached from $s$ on a path of cost~$0$.
The agent will definitely travel from $s$ to $t$ and not abandon if and only if $t$ can be reached from every vertex of $V'$ on a path of cost~$0$.
Because no rewards need to be placed in this scenario, the optimal budget is $b=0$.
Next assume $\beta=1$.
When the agent is at $s$, its lowest perceived cost for moving from $s$ to $t$ is equal to $d(s)$.
Setting $r(t) = d(s)$ yields a motivating and also optimal reward configuration.
This holds true because the agent travels from $s$ to $t$ along the edges of a cheapest path and its lowest perceived cost do not increase on its walk. 
\end{proof}

We now formally introduce the decision problem MRC.

\begin{definition}[MOTIVATING REWARD CONFIGURATION]
Given a task graph $G$, a non-negative budget $b$ and bias factor $\beta\in[0,1]$, decide the existence of a motivating reward configuration $r$, with $r(v) \geq 0$ for all vertices of $G$, such that the total reward collected on any of the agent's walks is at most $b$.
\end{definition}

The following proposition establishes membership of MRC in NP.

\begin{proposition}\label{prop:3}
For any task graph $G$, reward configuration $r$ and bias factor $\beta$, it is possible to decided in polynomial-time if $r$ is motivating within a given budget of $b$.
\end{proposition}
\begin{proof}
Similar to Proposition~\ref{prop:2}, we modify $G$ in the following way.
First, we compute the lowest perceived cost $\zeta_r(v)$ for all vertices in $G$.
Next we take a copy of $G$, say $G'$, and remove all edges from $G'$ that do not minimize the agents perceived cost in the original graph $G$.
More formally, we remove all edges $(v,w)$ for which $\zeta_r(v) < c_r(c) + \beta(d_r(w)-r(w))$.
The given reward configuration $r$ is motivating if and only if $\zeta_r(v) \leq 0$ for all $v$ that can be reached from $s$ in $G'$.
This condition can be checked in polynomial-time using graph search algorithms.
To determine if the budget constraint is satisfied, we assign all edges $(v,w)$ of $G'$ a cost of $r(w)$.
Let $c$ be the maximum cost among all paths from $s$ to $t$ in $G'$ according to these prices.
Because $G'$ is acyclic, $c$ can be computed in polynomial-time.
Thus $r$ is within budget if and only if $c + r(s) \leq b$.
\end{proof}

Our NP-completeness proof of MRC relies on a reduction form SET PACKING (SP)~\cite{GJ}.
For convenience, we restate SP in the following Definition.

\begin{definition}[SET PACKING]
Given a collection $S_1, \ldots, S_\ell$ of finite sets and an integer $k\leq \ell$, decide if there exist at least $k$ mutually disjoint sets?.
\end{definition}

We are now ready to prove NP-completeness of MRC.

\begin{theorem}\label{th:comp2}
MRC is NP-complete, for any bias factor $\beta \in (0,1)$, even if $b=0$.
\end{theorem}
\begin{proof}
By Proposition~\ref{prop:3}, we can take any motivating reward configuration $r$ within budge $b$ as certificate for a "yes"-instance of MRC.
Hence MRC is in NP.
In the following we present a polynomial-time reduction from SP to MRC.
We focus on the case that $b=0$.
At the end of the proof we show how to modify the reduction to handle arbitrary values $b>0$. 

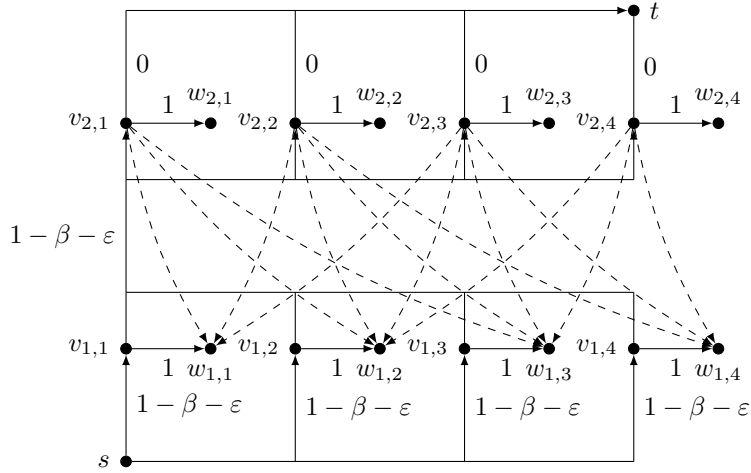
\begin{figure}[t]
\center
		\begin{tikzpicture}[nst/.style={draw,circle,fill=black,minimum size=4pt,inner sep=0pt]}, est/.style={draw,>=latex,->}]
			
			\node[nst] (v1) at (0,0) [label=left:\small{$s$}] {};	
			\node[nst] (v2) at (0,1.5) [label=left:\small{$v_{1,1}$}] {};
			\node[nst] (v3) at (2.25,1.5) [label=left:\small{$v_{1,2}$}] {};
			\node[nst] (v4) at (4.5,1.5) [label=left:\small{$v_{1,3}$}] {};
			\node[nst] (v5) at (6.75,1.5) [label=left:\small{$v_{1,4}$}] {};
			\node[nst] (v6) at (0,4.5) [label=left:\small{$v_{2,1}$}] {};
			\node[nst] (v7) at (2.25,4.5) [label=left:\small{$v_{2,2}$}] {};
			\node[nst] (v8) at (4.5,4.5) [label=left:\small{$v_{2,3}$}] {};
			\node[nst] (v9) at (6.75,4.5) [label=left:\small{$v_{2,4}$}] {};
			\node[nst] (v10) at (6.75,6) [label=right:\small{$t$}] {};
			\node[nst] (v11) at (1.125,1.5) [label=below:\small{$w_{1,1}$}] {};
			\node[nst] (v12) at (3.375,1.5) [label=below:\small{$w_{1,2}$}] {};
			\node[nst] (v13) at (5.625,1.5) [label=below:\small{$w_{1,3}$}] {};
			\node[nst] (v14) at (7.875,1.5) [label=below:\small{$w_{1,4}$}] {};
			\node[nst] (v15) at (1.126,4.5) [label=above:\small{$w_{2,1}$}] {};
			\node[nst] (v16) at (3.375,4.5) [label=above:\small{$w_{2,2}$}] {};
			\node[nst] (v17) at (5.625,4.5) [label=above:\small{$w_{2,3}$}] {};
			\node[nst] (v18) at (7.875,4.5) [label=above:\small{$w_{2,4}$}] {};
			
			
			\path (v1) edge (6.75,0)
			(v1) edge[est] node [right] {\small{$1-\beta-\varepsilon$}} (v2)
			(2.25,0) edge[est] node [right] {\small{$1-\beta-\varepsilon$}} (v3)
			(4.5,0) edge[est] node [right] {\small{$1-\beta-\varepsilon$}} (v4)
			(6.75,0) edge[est] node [right] {\small{$1-\beta-\varepsilon$}} (v5)
			(v2) edge (0,2.25)			
			(v3) edge (2.25,2.25)
			(v4) edge (4.5,2.25)
			(v5) edge (6.75,2.25)
			(0,2.25) edge (6.75,2.25)
			(0,2.25) edge node [left] {\small{$1-\beta-\varepsilon$}} (0,3.75)
			(0,3.75) edge (6.75,3.75)
			(0,3.75) edge[est] (v6)
			(2.25,3.75) edge[est] (v7)
			(4.5,3.75) edge[est] (v8)
			(6.75,3.75) edge[est] (v9)
			(v6) edge node [right] {\small{$0$}} (0,6)
			(v7) edge node [right] {\small{$0$}} (2.25,6)
			(v8) edge node [right] {\small{$0$}} (4.5,6)
			(v9) edge node [right] {\small{$0$}} (v10)
			(0,6) edge[est] (v10)
			(v2) edge[est] node [below] {\small{$1$}} (v11)
			(v3) edge[est] node [below] {\small{$1$}} (v12)
			(v4) edge[est] node [below] {\small{$1$}} (v13)
			(v5) edge[est] node [below] {\small{$1$}} (v14)
			(v6) edge[est] node [above] {\small{$1$}} (v15)
			(v7) edge[est] node [above] {\small{$1$}} (v16)
			(v8) edge[est] node [above] {\small{$1$}} (v17)
			(v9) edge[est] node [above] {\small{$1$}} (v18)
			(v6) edge[est,bend right=10, dashed] (v11)
			(v6) edge[est,bend right=10, dashed] (v12)
			(v6) edge[est,bend right=10, dashed] (v13)
			(v7) edge[est,bend left=10, dashed] (v11)
			(v7) edge[est,bend right=10, dashed] (v12)
			(v7) edge[est,bend right=10, dashed] (v13)
			(v7) edge[est,bend right=10, dashed] (v14)
			(v8) edge[est,bend left=10, dashed] (v11)
			(v8) edge[est,bend left=10, dashed] (v12)
			(v8) edge[est,bend right=10, dashed] (v13)
			(v8) edge[est,bend right=10, dashed] (v14)
			(v9) edge[est,bend left=10, dashed] (v12)
			(v9) edge[est,bend left=10, dashed] (v13)
			(v9) edge[est,bend right=10, dashed] (v14);

		\end{tikzpicture}
\caption{The task graph for $S_1=\{a,c,d\}$, $S_2=\{a,b\}$, $S_3=\{b,c,e\}$, $S_4 = \{b,e\}$ and $k=2$.}\label{fig:sp}
\end{figure}

Let ${\cal I}$ be an arbitrary instance of SP, consisting of finite sets $S_1, \ldots, S_\ell$ and an integer $k\leq \ell$.
We start by constructing the task graph $G$ of an MRC instance ${\cal J}$. 
Figure~\ref{fig:sp} depicts $G$ for a small sample instance of SP. 
In general, $G$ consists of a source $s$, a target $t$ and $k\ell$ vertices $v_{i,j}$, where $1\leq i \leq k$ and $1\leq j\leq \ell$. 
Intuitively, if the agents visit $v_{i,j}$, then $S_j$ is the $i$-th set to be added to the solution of $\mathcal{I}$. 
For every $v_{i,j}$, with $i<k$, there is a directed edge to all vertices $v_{i+1,j'}$ on the next level. 
We call these edges {\em upward edges\/}. 
The cost of any such edge is $1-\beta -\varepsilon$. 
Here $\beta\in (0,1)$ is any value whose encoding length is polynomial in that of ${\cal I}$. 
Furthermore, $\varepsilon$ is a positive value satisfying 
\[\varepsilon < \min\biggl\{\frac{(1-\beta)^2}{k}, \frac{\beta-\beta^2}{k-1+\beta}\biggr\}.\]
In Figure~\ref{fig:sp} the upward edges are merged to maintain readability. 
From $s$ there is a directed edge to every vertex $v_{1,j}$ on the bottom level.
Again, the cost of all such edges is $1-\beta-\varepsilon$.
Finally, for every vertex $v_{k,j}$ on level $k$ there exists an edge of cost~$0$ to $t$. 

In order guide the agent onto a specific upward edge, we add {\em shortcuts\/} to $G$ that connect every $v_{i,j}$ to $t$ via an intermediate vertex $w_{i,j}$.
The first edge $(v_{i,j},w_{i,j})$ has cost~$1$ and the second edge $(w_{i,j},t)$ has cost~$0$. 
In Figure~\ref{fig:sp} the edges $(w_{i,j},t)$ are omitted for the sake of a readability.
As we will see, a reward of value less than~$1/\beta$ can be placed on $w_{i,j}$ and the agent will not claim it.
Finally we introduce {\em downward paths\/} of length two that connect each $v_{i,j}$ with all $w_{i',j'}$ for which $i'<i$ and $S_j\cap S_{j'}\neq \emptyset$, i.e.\ the sets are not disjoint.
The first edge has a cost of~$0$, while the second edge has a cost of $(1-\beta-k\varepsilon)/(\beta-\beta^2)$. 
In Figure~\ref{fig:sp} each downward path is drawn as a single dashed edge.
The purpose of these paths is to let the agent claim a reward or abandon whenever the disjointness constraints of ${\cal I}$ are violated.
Notice that $G$ is acyclic.
We will show that ${\cal I}$ has a solution if and only if ${\cal J}$ has one.

($\Longrightarrow$) First assume that ${\cal I}$ has a solution, i.e.\ there exist $k$ mutually disjoint sets among $S_1, \ldots, S_\ell$.
Fix such a selection of $k$ disjoint sets and assign each set to a distinct level of $G$. 
Suppose that $S_j$ of the collection is assigned to level $i$. 
Then place a reward of value $(1-\varepsilon)/\beta$ on $w_{i,j}$.
The corresponding shortcut from $v_{i,j}$ is referred to as an {\em active shortcut\/}. 
We now analyze the agent's walk through $G$ and show that it visits exactly the vertices $v_{i,j}$ at which the active shortcuts start.
The agent does not claim any reward.

Suppose that the agent is located at the initial vertex $v_{i,j}$, with $i < k$, of an active shortcut.
There are three options.
First the agent could follow the shortcut to $w_{i,j}$. 
However, the perceived cost along this path to $t$ is $\varepsilon$, so the agent has no incentive to move to $w_{i,j}$.
Secondly, the agent could take a downward path.
By construction, none of them leads to an active shortcut.
This means that the agent cannot collect any reward on such a path but encounters a positive cost on the downward path.
Hence this option is not motivating either.
The agent's only remaining option is to take an upward edge. 
If the agent plans to take the active shortcut of level $i+1$, it perceives a total cost of $0$. 
This is a motivating choice. 
Conversely, assume that the agent plans a path $P$ to $t$ that visits a vertex $v_{i+1,j'}$ on the next level such that the corresponding shortcut is not active.
We distinguish between four scenarios.
If $P$ includes a downward path, then at best one reward can be located along $P$ so that the agent's perceived cost is at least 
\[1-\beta -\varepsilon + \beta\biggl(\frac{1-\beta-k\varepsilon}{\beta-\beta^2} - \frac{1-\varepsilon}{\beta}\biggr) = \frac{k}{1-\beta}\biggl(\frac{(1-\beta)^2}{k} - \varepsilon\biggr) >0,\]
which is not motivating.
The inequality follows from our choice of $\varepsilon$.
If $P$ includes the shortcut at $v_{i+1,j'}$ then the agent encounters no reward but a positive cost, which is not motivating either.
If $P$ includes a shortcut on some level above $i+1$, then the agent must traverse at least two upward edges but can collect at most one reward.
In this case the agent's perceived cost is lower bounded by
\[1-\beta -\varepsilon + \beta \biggl((1-\beta -\varepsilon) + 1 - \frac{1-\varepsilon}{\beta}\biggr) = \beta(1-\beta-\varepsilon) > \beta\biggl(\frac{(1-\beta)^2}{k}-\varepsilon\biggr) >0.\]
As always, the last inequality follows by choice of $\varepsilon$.
Finally, $P$ may neither include a downward path nor a shortcut.
However, this means that the agent has positive cost and does not collect any reward.
All in all, the only motivating option is to take the upward edge leading to the active shortcut of level $i+1$.

The same arguments also apply when the agent is at $s$ or the initial vertex of an active shortcut on the top level. 
At $s$, the agents only option is to take an upward edge.
Hence it moves to the initial vertex of the active shortcut of the bottom level.
At the top level the agent will take the direct edge to $t$, which incurs no cost.
All other options, namely taking a downward path or the current shortcut, are not motivating.

($\Longleftarrow$) Next assume that ${\cal J}$ has a solution, i.e.\ there exists a motivating reward configuration such that the agent does not claim any reward.
Consider arbitrary walks of the agent.
A first crucial observation is that no such walk enters a shortcut or a downward path, because a positive reward on a vertex along these paths is needed to guide the agent onto them.
Considering that the agent cannot change its plan once it entered a shortcut or downward path, it would either claim the reward or abandon, which contradicts the assumption that ${\cal J}$ has a solution. 
Hence the agents visits one vertex $v_{i,j}$ at each level $i$. 
We call every $v_{i,j}$ that is contained in one of the agent's possible walks an {\em active vertices\/}.
Note that there might be more than one active vertex per level.

We next prove that at every active vertex the agent's lowest perceived cost at least $(1-k)\varepsilon$. 
More specifically, we show by backwards induction, from the top level down to the bottom level, that whenever the agent is located at an active vertex of level~$i$, its perceived cost in planing a path to $t$ is at least $(i-k)\varepsilon$. 
Moreover, if $i<k$, then the only motivating paths are paths passing through the shortcuts of active vertices on level~$i+1$. 
First assume that the agent is at an active vertex $v_{k,j}$ on the top level. 
As argued above, the agent cannot take the shortcut or downward path to $t$. 
However, the edge $(v_{k,j},t)$ is a motivating path with a perceived cost of $0$, which is equal to $(k-k)\varepsilon$.
This proves the basis of our induction.

For the inductive step, suppose that $i<k$ and that the agent is located at some active vertex $v_{i,j}$.
Let $v_{i+1,j'}$ be the active vertex visited next by the agent. 
Because the agent will move from $v_{i,j}$ to $v_{i+1,j'}$, there must exist a path $P$ from $v_{i,j}$ to $t$ via $(v_{i,j},v_{i+1,j'})$ that minimizes the agent's perceived cost.
We distinguish four scenarios.
First, assume that $i = k-1$  and $P$ contains $(v_{i+1,j'}, t)$.
This means that the agent receives no reward but has positive cost, which is not motivating.
Secondly, assume $P$ contains a forward edge leaving $v_{i+1,j'}$ and consider the perceived cost of the remaining portion of $P$ when viewed from $v_{i+1,j'}$.
By induction hypothesis this cost must be at least $((i+1)-k)\varepsilon$.
Furthermore, no reward must be placed at $v_{i+1,j'}$, as this would violate the budget.
This means that the perceived cost of $P$ at $v_{i,j}$ increases by $\beta (1-\beta -\varepsilon)$ when compared to the perceived cost of $P$ at $v_{i+1,j'}$.
Thus the perceived cost at of $P$ at $v_{i,j}$ is at least 
\[\bigl((i+1)-k\bigr)\varepsilon + \beta (1-\beta -\varepsilon) = (k-(i+1)+\beta)\biggl(\frac{\beta-\beta^2}{k-(i+1)+\beta}-\varepsilon\biggr)>(k-(i+1)+\beta)\biggl(\frac{\beta-\beta^2}{k-1+\beta}-\varepsilon\biggr)>0.\]
The last inequality holds by choice of $\varepsilon$.
Hence $P$ is not motivating.
Thirdly, assume $P$ contains a downward path out of $v_{i+1,j'}$.
In this case, the perceived cost of $P$ at $v_{i,j}$ compared to the perceived cost of $P$ at $v_{i+1,j'}$ increases by even more, namely $1-\beta -\varepsilon$.
Certainly, $P$ can not be motivating.  
Finally, assume $P$ contains the shortcut out of $v_{i+1,j'}$.
When viewed from $v_{i,j}$ instead of $v_{i+1,j'}$, the perceived cost of $P$ increases by $1-\beta -\varepsilon$ and decreases by $1-\beta$.
Thus the perceived cost is at least $((i+1)-k)\varepsilon -\varepsilon = (i-k)\varepsilon$ which concludes the induction step.
By a similar argument, the only motivating paths out of $s$ traverse the shortcut of an active vertex on the bottom level.

The last three paragraphs imply that for every active vertex $v_{i,j}$ a reward of at least $(1-\varepsilon)/\beta$ has to be placed at $w_{i,j}$.
Otherwise the shortcut would not be motivating when the agent resides at an active vertex on the previous level $i-1$, or at $s$ if $i=1$.
This implies that there can be no downward path connecting an active vertex $v_{i,j}$ to the shortcut of an active vertex on a lower level, because the perceived cost at $v_{i,j}$ for following the downward path would be at most
\[\beta\biggl(\frac{1-\beta-k\varepsilon}{\beta-\beta^2} - \frac{1-\varepsilon}{\beta}\biggr) = \frac{(1-k-\beta)\varepsilon}{1-\beta} < (1-k)\varepsilon.\]
Thus the active vertices $v_{i,j}$ along any walk of the agent correspond to $k$ disjoint sets $S_j$, which proves that ${\cal I}$ has a solution.

We finally address the case that the agent may collect a total reward of $b>0$. 
Consider a slightly modify version of $G$.
We rename $t$ by $t'$ and add an edge from $t'$ to a new target $t$. 
The cost of this edge is $\beta b$. 
The agent only reaches $t$ from $t'$ if a reward of value $b$ is placed on $t$. 
With this observation the above proof immediately carries over.
\end{proof}

We next turn to the optimization variant of MRC. 

\begin{definition}[MOTIVATING REWARD CONFIGURATION OPT]
Given a task graph $G$ and a bias factor $\beta\in(0,1)$, determine the minimum budget $b$ for which there exists a reward configuration $r$, with $r(v) \geq 0$ for all vertices of $G$,such that the total reward collected on any of the agent's walks is at most $b$.
\end{definition}

Assuming that ${\rm P}\neq{\rm NP}$, Theorem~\ref{th:comp2} implies that there exists no polynomial-time algorithm that approximates a motivating reward configuration such that the required budget is within any ratio greater of equal to $1$ compared to the budget of an optimal solution.
This, follows from the fact that MRC is NP-complete even in the special case that $b=0$.

\begin{corollary}\label{cor:approx}
MRC-OPT is NP-hard to approximated within any ratio greater of equal to $1$.
\end{corollary}

\section{Conclusions}
In this paper we have studied computational problems in time-inconsistent planning using a graph model by Kleinberg and Oren~\cite{KO}. 
As a main result, assuming ${\rm P}\neq{\rm NP}$, we established asymptotically tight upper and lower bounds of $\Theta(\sqrt{n})$ on the efficient approximability of MSG-OPT as well as a negative approximability result for MRC-OPT.
Given the state of the art, we believe that a generalization of the graph model to quasi hyperbolic discount functions is a promising research direction. 
In hyperbolic discounting, which is frequently used in the behavioral economics literature~\cite{FLO,La}, there are two parameters $\beta,\delta\in [0,1]$.
Any value $c$ that is realized $t$ time steps in the future is perceived with a current value of $\beta \delta^t c$. 
For $t=0$, the perceived value is $c$.
Note that Kleinberg and Oren's~\cite{KO} model is a special case of quasi hyperbolic discounting for $\delta=1$.
Although such discount functions are more involved, the exponentially fading value of future costs and rewards might allow for improved approximation guarantees if $\delta<1$.

\section*{Appendix}

\begin{proposition}
$k$-DCP is NP-complete in directed acyclic graphs.
\end{proposition}
\begin{proof}
Membership of $k$-DCP in NP is easy to see.
In order to show NP-hardness we modify Lynch's~\cite{L} reduction that mapped instances of 3-SAT to instances of $k$-DCP in undirected graphs.
Let $\mathcal{I}$ be a $3$-SAT instance with $m$ clauses $c_1 \ldots c_m$ over $n$ variables $x_1 \ldots x_n$.
Furthermore, let $\mathcal{J}$ be an $(m+n)$-DCP instance on a directed acyclic graph $H$ that is constructed in the following way.
For every variable $x_i$, there are terminals $s_i$ and $t_i$ that are connected via two vertex-disjoint paths.
One path, the so called high path, corresponds to the case that $x_i$ is set to true, while the other path, called low path, corresponds to the case that $x_i$ is false.
Similar to the variables, there are terminals $s'_j$ and $t'_j$ for each clause $c_j$.
Again, the terminals are connected by vertex-disjoints paths, this time one for each literal in $c_j$.
If the $k$-th literal of $c_j$ is equal to $x_i$, a vertex $v_{i,j,k}$ is added to the low path of $x_i$ and the literal's path.
Should the literal be negated, $v_{i,j,k}$ is added to the high path of $x_i$ and the literal's path instead.
Assuming that the position of the vertices $v_{i,j,k}$ along their respective paths is ordered according to the indices $i$ and $j$, one obtains a directed acyclic graph $H$.
Moreover, it is easy to verify that the size of $H$ is polynomial with respect to the original $3$-SAT formula.
Figure~\ref{fig:kdcp} shows an example of $H$ for a small sample instance of 3-SAT.
We next prove that $\mathcal{I}$ is feasible if and only if $\mathcal{J}$ is.

\begin{figure}[t]
\center	
\begin{tikzpicture}[nst/.style={draw,circle,fill=black,minimum size=4pt,inner sep=0pt]}, est/.style={draw,>=latex,->}]
			
				\node[nst] (v1) at (0,5.5) [label=left:\small{$s_1$}] {};
				\node[nst] (v2) at (0,3.5) [label=left:\small{$s_2$}] {};
				\node[nst] (v3) at (0,1.5) [label=left:\small{$s_3$}] {};
				\node[nst] (v4) at (1.5,7) [label=above:\small{$s'_1$}] {};
				\node[nst] (v5) at (4,7) [label=above:\small{$s'_2$}] {};
				\node[nst] (v6) at (7,7) [label=above:\small{$s'_3$}] {};
				\node[nst] (v7) at (9,5.5) [label=right:\small{$t_1$}] {};
				\node[nst] (v8) at (9,3.5) [label=right:\small{$t_2$}] {};
				\node[nst] (v9) at (9,1.5) [label=right:\small{$t_3$}] {};
				\node[nst] (v10) at (1.5,0) [label=below:\small{$t'_1$}] {};
				\node[nst] (v11) at (4,0) [label=below:\small{$t'_2$}] {};
				\node[nst] (v12) at (7,0) [label=below:\small{$t'_3$}] {};
				\node[nst] (v13) at (3,5) [label=above left:\small{$v_{1,2,1}$}] {};
				\node[nst] (v14) at (6,6) [label=above left:\small{$v_{1,3,1}$}] {};
				\node[nst] (v15) at (1,4) [label=above left:\small{$v_{2,1,1}$}] {};
				\node[nst] (v16) at (4,3) [label=above left:\small{$v_{2,2,2}$}] {};
				\node[nst] (v17) at (7,3) [label=above left:\small{$v_{2,3,2}$}] {};
				\node[nst] (v18) at (2,2) [label=above left:\small{$v_{3,1,2}$}] {};
				\node[nst] (v19) at (5,2) [label=above left:\small{$v_{3,2,3}$}] {};
				\node[nst] (v20) at (8,1) [label=above left:\small{$v_{3,3,3}$}] {};
				
				\path (v1) edge (0,6)
				(0,6) edge[est] (v14)
				(v14) edge (9,6)
				(9,6) edge[est] (v7)
				(v1) edge (0,5)
				(0,5) edge[est] (v13)
				(v13) edge (9,5)
				(9,5) edge[est] (v7)
				(v2) edge (0,4)
				(0,4) edge[est] (v15)
				(v15) edge (9,4)
				(9,4) edge[est] (v8)
				(v2) edge (0,3)
				(0,3) edge[est] (v16)
				(v16) edge[est] (v17)
				(v17) edge (9,3)
				(9,3) edge[est] (v8)
				(v3) edge (0,2)
				(0,2) edge[est] (v18)
				(v18) edge[est] (v19)
				(v19) edge (9,2)
				(9,2) edge[est] (v9)
				(v3) edge (0,1)
				(0,1) edge[est] (v20)
				(v20) edge (9,1)
				(9,1) edge[est] (v9)
				(v4) edge (1,7)
				(1,7) edge[est] (v15)
				(v15) edge (1,0)
				(1,0) edge[est] (v10)
				(v4) edge (2,7)
				(2,7) edge[est] (v18)
				(v18) edge (2,0)
				(2,0) edge[est] (v10)
				(v5) edge (3,7)
				(3,7) edge[est] (v13)
				(v13) edge (3,0)
				(3,0) edge[est] (v11)
				(v5) edge[est] (v16)
				(v16) edge[est] (v11)
				(v5) edge (5,7)
				(5,7) edge[est] (v19)
				(v19) edge (5,0)
				(5,0) edge[est] (v11)
				(v6) edge (6,7)
				(6,7) edge[est] (v14)
				(v14) edge (6,0)
				(6,0) edge[est] (v12)
				(v6) edge[est] (v17)
				(v17) edge[est] (v12)
				(v6) edge (8,7)
				(8,7) edge[est] (v20)
				(v20) edge (8,0)
				(8,0) edge[est] (v12);
				
\end{tikzpicture}
\caption{$H$ for $(\bar{x}_2\vee\bar{x}_3)\wedge(x_1\vee x_2 \vee \bar{x}_3)\wedge(\bar{x}_1\vee x_2 \vee x_3)$}
\label{fig:kdcp}
\end{figure}
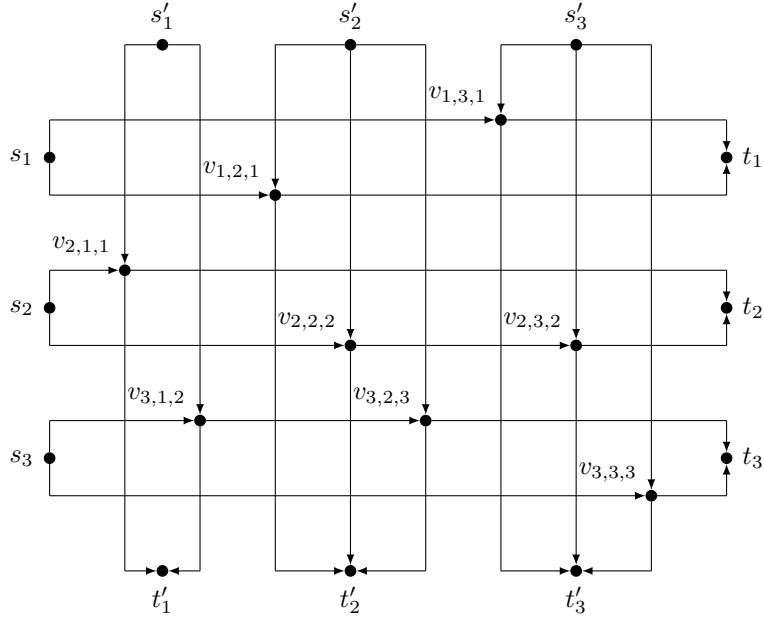

($\Longleftarrow$) Given a variable assignment that satisfies the $3$-SAT formula, it is easy to construct $m+n$ vertex-disjoint connecting path.
For all pairs $s_i$ and $t_i$, choose the high path if $x_i$ is true or the low path if $x_i$ is false.
For $s'_j$ and $t'_j$ the path of any literal in $c_j$ that evaluates to true can be selected.
Because the formula is satisfied, at least one such path must exist.
By construction of $H$, this yields $m+n$ vertex disjoint connecting path.
		
($\Longrightarrow$) Observe, that every literal's path has exactly one intermediate vertex.
As a result, the only paths that connect a terminal $s'_j$ with $t'_j$ are exactly the paths that correspond to the literals of $c_j$.
The same holds for the high and low paths of a variable $x_i$.
Hence, if there are $m+n$ vertex-disjoint connecting path in $H$, then the chosen high and low paths directly translate to a satisfying variable assignment of the $3$-SAT formula.
\end{proof}

\end{document}